\documentclass[aps,prx,twocolumn,superscriptaddress]{revtex4-2}
\usepackage[utf8]{inputenc}

\usepackage[titletoc,toc,title,page]{appendix}

\usepackage{csquotes,graphicx}
\usepackage[italian,english]{babel}

\usepackage{microtype} 

\usepackage{bm} 

\usepackage{dsfont} % provides more \mathbb \smallskipymbols
\usepackage{amsmath,amssymb,amsthm,thmtools}
\usepackage{mathtools}
\usepackage{cases}
\usepackage{calc,palatino}
\usepackage{mathrsfs} % provides the nice \matchsrc font
\usepackage[normalem]{ulem} %provides \sout command for striking through text

\usepackage{parskip}
\usepackage{nameref}
\usepackage[colorlinks=true]{hyperref}
\usepackage[nameinlink]{cleveref}
\crefname{appsec}{Appendix}{Appendices}
\crefname{box}{Box}{Box}
\hypersetup{
  colorlinks   = true, %Colours links instead of ugly boxes
  urlcolor     = green!80!black, %Colour for external hyperlinks
  linkcolor    = blue, %Colour of internal links 	q
  citecolor    = red!80!black %Colour of citations
}

\usepackage{physics} % provides Dirac notation and lots of other things

\usepackage{float} % provides the option H for includegraphics

\usepackage{graphicx}
\usepackage{subcaption}

\usepackage[usenames,dvipsnames,table]{xcolor}
\usepackage{easyReview}

\usepackage{tikz}
\usetikzlibrary{calc,shapes.geometric}

\usepackage{placeins}
\usepackage{multirow,tabularx,booktabs}
\usepackage[most]{tcolorbox}
\newtcbtheorem{tbox}{Box}{enhanced, float*=t, width=\textwidth, label type=box}{box}

\usepackage[printonlyused,withpage,nohyperlinks,smaller]{acronym}

\graphicspath{{./figures/}}

\newcommand{\calE}{\mathcal{E}}

\newcommand{\calN}{\mathcal{N}}

\newcommand{\calL}{\mathcal{L}}

\newcommand{\calZ}{\mathcal{Z}}

\newtheorem*{app}{Theorem}

\begin{document}

\title{Broken Detailed Balance and Entropy Production in CPTP Quantum Brownian Motion}
\author{Simone Artini}
\let\comma,
\affiliation{Quantum Theory Group\comma{} Dipartimento di Fisica e Chimica Emilio Segr\`e\comma{} Universit\`a degli Studi di Palermo\comma{} via Archirafi 36\comma{} I-90123 Palermo\comma{} Italy}
\author{Gabriele Lo Monaco}
\let\comma,
\affiliation{Quantum Theory Group\comma{} Dipartimento di Fisica e Chimica Emilio Segr\`e\comma{} Universit\`a degli Studi di Palermo\comma{} via Archirafi 36\comma{} I-90123 Palermo\comma{} Italy}
\author{Alberto Imparato}
\affiliation{Dipartimento di Fisica\comma{} Universit\`a degli Studi di Trieste\comma{} Strada Costiera 11\comma{} 34151 Trieste\comma{} Italy}
\affiliation{ Istituto Nazionale di Fisica Nucleare\comma{} Trieste Section\comma{} Via Valerio 2\comma{} 34127 Trieste\comma{} Italy}

\author{Mauro Paternostro}
\let\comma,
\affiliation{Quantum Theory Group\comma{} Dipartimento di Fisica e Chimica Emilio Segr\`e\comma{} Universit\`a degli Studi di Palermo\comma{} via Archirafi 36\comma{} I-90123 Palermo\comma{} Italy}
\affiliation{Centre for Quantum Materials and Technologies\comma{} School of Mathematics and Physics\comma{} Queen’s University Belfast\comma{} BT7 1NN\comma{} United Kingdom}

\author{Sandro Donadi}
\affiliation{Dipartimento di Ingegneria\comma{} Universit\`a degli Studi di Palermo\comma{} via Archirafi 36\comma{} I-90123 Palermo\comma{} Italy}

\begin{abstract}
We rigorously analyze the non-equilibrium thermodynamic behavior of various formulations of quantum Brownian motion (QBM) using the framework of stochastic thermodynamics. 
While the widely used Caldeira-Leggett master equation exhibits desirable thermodynamic features — such as the fulfilment of a detailed balance — it fails to
ensure complete positivity. In contrast, several completely positive and trace-preserving (CPTP) extensions turn out to be thermodynamically controversial. 
We show that such extensions introduce anomalous phase-space structures that violate detailed balance at the steady state, leading to non-vanishing entropy production and spurious non-equilibrium currents lacking  a clear physical interpretation. Our results highlight a fundamental tension between quantum consistency and thermodynamic equilibration in open quantum systems.
\end{abstract}

\maketitle
 
\section{Introduction}
The theory of Brownian motion is a cornerstone in modern physics. It played a decisive role in the consolidation of the atomic theory of matter~\cite{einstein1905movement} and, from a statistical mechanics point of view, it served as a starting point for the study of continuous Markov processes and their role in describing the mechanism that leads mesoscopic systems to thermal equilibrium, and links it to dissipation \cite{kubo1966fluctuation, onsager1931reciprocal}.
The interaction between a particle and a thermal environment, the description of dissipation and the emergence of irreversibility are also of profound interest in the context of quantum mechanics, both from a fundamental and a technological perspective. The quantum description of Brownian motion was first introduced in the seminal work of Caldeira and Leggett~\cite{caldeira1981influence}. Starting from a model of a bath of harmonic oscillators, they derived a master equation that correctly predicts the relaxation of the system's state to a thermal state at the temperature of the bath. However, the dynamics described by the Caldeira-Leggett (CL) master equation are not of Lindblad form and do not preserve complete positivity (CP) (to be more precise, it does not preserve even simple positivity~\cite{homa2019positivity}). While the violation of CP is compatible with the approximations made in deriving the CL master equation—and thus does not lead to inconsistencies within its regime of validity—it is desirable to formulate an open quantum system dynamics that preserves complete positivity. In order to cast the master equation in Lindblad form, the addition of a term that is small in the high temperature limit is proposed \cite{breuer2002theory}. In order to go beyond arbitrary approaches and corrections, many works expanded on the original CL model either by relaxing the approximations, thus abandoning the Markovian framework~\cite{hu1992quantum, hu1993quantum, vacchini2007relaxation, giovannetti2001phase}, or by implementing phenomenological Lindblad master equations based on global symmetries \cite{vacchini2002quantum}.

In this work, we show that the inclusion of the extra term in the CL model, while ensuring consistency with CP, implies that the steady state of the system is no longer an equilibrium one. In fact, the correction leads to the presence of persistent currents in the steady state, signaling a violation of the detailed balance condition. Thus, enforcing complete positivity introduces anomalous currents that prevent proper thermalization. 
We then generalize our analysis to the most general Gaussian Lindblad master equation that is translationally covariant in space, finding again a violation of detailed balance by the steady state. Therefore, within the framework of Gaussian and Lindblad dynamics, the quantum description of a system coupled to a thermal bath that is invariant under translations fundamentally differs from the classical picture: preserving complete positivity—a genuinely quantum requirement—implies the presence of effective driving mechanisms that do not allow the steady state to be in true equilibrium. Finally, we show that restoring a proper equilibration dynamics is indeed possible by breaking the symmetry thus far imposed, but requires a fine tuning on the parameters of the Lindbladian that depends on the Hamiltonian of the system.

In relation to the existing literature on quantum open systems, our work provides a novel and distinct perspective on the fundamental interplay between dynamical and thermodynamic consistency. While previous studies have addressed the challenges of Lindblad descriptions for quantum systems in thermal environments, their focuses differ significantly from ours. For instance, discussions regarding the steady-state discrepancies at low temperatures or under changing Hamiltonians often rely on enforcing detailed balance (DB) in the Kubo-Martin-Schwinger (KMS) sense to guarantee a Gibbs equilibrium state \cite{stockburger2017thermodynamic}. Conversely, other approaches have explored how standard KMS DB conditions fail when the steady-state is non-diagonal in the energy eigenbasis, prompting the introduction of generalized DB definitions based on hidden time-reversal symmetries for driven-dissipative systems \cite{roberts2021hidden}. Rather than modifying the core definition of detailed balance, our manuscript adheres to the standard, widely accepted formulation of DB for classical markovian dynamics-when consistent with the quantum setting-to rigorously examine the compatibility between the requirement of the dynamics to be CPTP and the steady-state (non-)equilibrium properties. Although recent advancements demonstrate that CPTP and DB can be reconciled by relinquishing spatial translation covariance and meticulously fine-tuning parameters \cite{nicacio2024complete}, our work systematically incorporates these edge cases while characterizing a broader, systemic tension between CPTP and DB.  Moreover, in this framework we also provide a comprehensive discussion on entropy production rates and their role in determining steady-state behaviors which existing literature largely lacks.

The remainder of this manuscript is organized as follows. In \cref{sec:GF}, we establish the general framework of stochastic thermodynamics for both classical and quantum systems, introducing the Wigner representation and formalizing the mathematical relationship between the DB condition, irreversible currents, and entropy production rates. In \cref{sec:results}, we present our main results by applying this framework to general translation-covariant, Gaussian, and CPTP master equations, demonstrating how enforcing complete positivity systematically violates detailed balance in the steady state. \Cref{sec:brokenTC} explores how breaking the spatial translation covariance symmetry can restore proper thermalization, analyzing the strict, Hamiltonian-dependent fine-tuning required for the Lindbladian parameters. Finally, \cref{sec:concl} provides concluding remarks and outlines future research avenues, particularly concerning non-Gaussian and non-Markovian open quantum systems.

\section{General framework}
\label{sec:GF}
There are three, nearly equivalent, mathematical descriptions of Brownian motion: the Langevin equation, which describes the stochastic evolution of the paths of the particle by adding noise to the deterministic equations of motion; stochastic differential equations (SDE), which are 
a mathematical refinement and generalization of the Langevin equation; and the Fokker-Planck equation (FP), which describes the evolution of the probability distribution of the position (and possibly momentum or velocity) of the particle~\cite{gardiner1985handbook}. In this work, we focus on the latter description, and consider the Fokker-Planck equation
\begin{equation} \label{eq: general_FP_main}
    \partial_t P(\bm x,t)= \left(-\nabla\!\cdot\! {\boldsymbol{A}}\boldsymbol{x}
    +\frac{1}{2} \nabla\!\cdot\!\boldsymbol{B}\nabla\right)P(\bm x,t) \,,
\end{equation}
where ${\bm x:=(\bm q_1,...\bm q_N,\bm p_1...\bm p_N)}$, $\boldsymbol{A}$ is the drift matrix and $\boldsymbol{A}\boldsymbol{x}$ the drift vector, which  contains the deterministic motion and friction terms; $\boldsymbol{B}$ is the diffusion matrix, which accounts for the stochastic part of the dynamics. The equilibrium condition is expressed by the notion of micro-reversibility in the steady state, meaning that each transition $\boldsymbol{x} \rightarrow \bm x'$ has the same probability of its time-reversed version. Time-reversal is performed  via a diagonal matrix ${\bm \calE}$ such that ${\bm \calE}\bm x = \tilde{\bm x}$, with $\tilde{x}_i = \epsilon_i x_i,\,\, \epsilon_i = + 1$ for positions and $\epsilon_i = - 1$ for momenta. Assuming the existence of a steady-state solution $P_S(\bm x)$, the micro-reversibility condition holds  
if and only if~\cite{gardiner1985handbook}
\begin{equation} \label{eq: cond_DB_main}
        \begin{aligned}
        &{\bm \calE}{\bm A}{\bm \calE}^T\bm x P_S(\bm x) = -{\bm A}\bm x P_S(\bm x) + {\bm B}\nabla_{\bm x}P_S(\bm x)\,,\\
         &{\bm \calE}{\bm B}{\bm \calE}^T = {\bm B},
    \end{aligned}
\end{equation}
and $P_S(\bm x)$ is said to satisfy the \textit{detailed balance} (DB) condition. A thermodynamic theory of Brownian motion, known as stochastic thermodynamics \cite{seifert2012stochastic, van2013stochastic}, has also been developed by interpreting usual thermodynamic quantities -namely entropy, work and heat- as random variables of the trajectories of the particle. In particular, this approach helped expanding thermodynamics to nonequilibrium systems and shed light onto the emergence of irreversibility and its connection to entropy \cite{seifert2008stochastic, tome2010entropy,Imparato7a}. This approach usually roots itself in the SDE description of Brownian motion, but the FP approach can also be used to find the average of the thermodynamic quantities defined on trajectories. The main quantity of interest is the entropy production rate, which, in many relevant scenarios, uniquely characterizes the distance from equilibrium \cite{campisi2011colloquium}. Starting from the Shannon entropy of the probability distribution $S(t) = -\int \,d\bm x P(\bm x,t)\ln P(\bm x,t)$, one can split the entropy rate of the system as
\begin{equation} \label{eq: phi_and_pi}
    \frac{dS(t)}{dt} = \Pi (t) - \Phi (t),
\end{equation}
where $\Phi (t)$ quantifies the rate at which entropy flows between system and environment, and $\Pi (t)$ is the entropy production rate, which characterizes thermodynamic irreversibility~\cite{LandiPaternostro}. Such splitting is usually performed by separating the drift and diffusion matrices into symmetric and antisymmetric components with respect to time inversion, thus identifying the irreversible and reversible parts of the current, respectively. The entropy production rate is then identified  with the term that is non-negative or null if and only if $P_S$ satisfies the DB condition, provided the diffusion has no odd component with respect to time inversion. This is generally written as 
\begin{equation}\label{eq: Pi}
    \Pi = 2\int \,d\bm x \,\, \frac{J^{I}(\bm x)\,({\bm B}^{I})^{-1}J^{I}(\bm x)}{P(\bm x)}\,,
\end{equation}
$J^{I}(\bm x)$ being the irreversible part of the current. This expression provides the entropy production rate averaged over the ensemble of stochastic individual trajectories of the dynamics, thus establishing a link with well-known fluctuation relations \cite{gardiner1985handbook,seifert2008stochastic,van2013stochastic,seifert2012stochastic,tome2010entropy}. 

In this paper we focus on exploring the relation between the DB condition and the irreversible currents in the stationary state. It is well known that, when the DB conditions hold, the stationary state exists, is unique and the irreversible current vanishes \cite{van1992stochastic}. However, the absence of an irreversible current in the steady state does not imply DB, as shown in the following theorem (we provide a proof in \cref{app: DB theorem}):

\begin{app}
    For a homogeneous Fokker-Planck equation such that 
    ${\mathcal{E}}{\bm B}{\mathcal{E}}^T={\bm B}$, if a stationary solution $P_S(\bm x)$ exists and is such that the irreversible component of the current is zero, then the DB conditions hold. 
\end{app}
These considerations justify the choice of the entropy production rate as the figure of merit to characterize nonequilibrium when the hypothesis of this Theorem is satisfied. Note that just having a zero irreversible current is \emph{not} enough to guarantee DB balance: indeed, having a reversible component in the diffusion matrix, i.e. ${\mathcal{E}}{\bm B}{\mathcal{E}}^T\neq{\bm B}$, leads to a reversible current in the steady state that also drives the system out of equilibrium.

Until now we discussed in the framework of classical thermodynamics, however our goal is to study models describing the Quantum Brownian Motion (QBM)~\cite{caldeira1981influence,breuer2002theory}. Restricting the study to potentials that are at most quadratic in the system's degrees of freedom, and considering only states with positive Wigner function

\begin{equation}
W(q, p) = \frac{1}{\pi\hbar} \int_{-\infty}^{\infty} \langle q+y | \hat{\rho} | q-y \rangle e^{-2ipy/\hbar} dy
\end{equation}

we can apply the same construction of classical continuous Markov processes to build the stochastic thermodynamics of some of the versions of the QBM model. This allows us to adopt the well-established framework from classical stochastic thermodynamics introduced in this Section. In doing so, we avoid having to choose among the many proposals for quantum detailed balance available in the literature \cite{fagnola2009two, alicki1976detailed, agarwal1973open}, none of which has yet been universally accepted as the definitive generalization of the classical notion.

To provide our quantitative assessment, we replace the Shannon entropy $S(t)$ with the Wigner entropy
\begin{equation}
S_{W}(t) = -\int \,d\bm x\, W_{\hat{\rho}}(\bm x,t)\ln(W_{\hat{\rho}}(\bm x,t))  \end{equation}
as the main figure of merit to characterize the thermodynamic properties of the considered models~\cite{brunelli2018experimental,brunelli2016irreversibility, santos2017wigner, santos2018irreversibility, artini2023characterizing, artini2025non}, where $W_{\hat{\rho}}(\bm x,t)$ is the Wigner function associated with the state $\hat{\rho}$ of a system. In \cref{app: wherl and Vn} we provide for comparison a discussion about other choices of entropy in this context and ~Ref.~\cite{colla2021entropy} for a thermodynamic description of the exact CL modelbased instead on the Von Neumann entropy.

Using this framework, we can identify anomalous terms in the FP equation for the Wigner function,  which do not have a counterpart in classical Brownian motion. These anomalous terms stem from the complete positivity requirement, and  affect the thermodynamics of the stationary state. In particular, we are able to demonstrate that all  translation-covariant and completely positive and trace preserving (CPTP) quadratic master equations lead to nonequilibrium steady states characterized by a constant entropy production rate or a reversible nonequilibrium current. As a consequence, thermalization, as intended classically, \emph{i.e.} satisfying the DB in \cref{eq: cond_DB_main}, cannot be achieved witin these models. We point out that the assumptions of the dynamics to be CPTP and translational covariance are very natural and invoke in several notable models: the Joos-Zeh and the Cladeira-Legget models \cite{joos1985emergence,caldeira1981influence} in the gaussian regime and more in general in equation of motion of optomechanical systems \cite{barchielli2015quantum}, the Continuous Spontaneous Collapse \cite{bassi2005energy} and the Diósi-Penrose \cite{diosi1987universal,diosi1989models,penrose1996gravity} models.

\section{Results} 
\label{sec:results}
We start by considering the most general Markovian CPTP and Gaussian dynamics. The requirement of the dynamics being CPTP and Markovian impose the Lindblad structure~\cite{gorini1976completely,lindblad1976generators,breuer2002theory} of the master equation for the statistical operator $\hat{\rho}$, while gaussianity means that the this equation is at most quadratic in $\hat{q}$ and $\hat{p}$. As discussed in detail in Ref.~\cite{vacchini2002quantum}, this  leads to a master equation of the form
\begin{equation}\label{eq: trasl_cov_me}
\begin{aligned}
\frac{d\hat{\rho}}{dt} ={}&
-\frac{i}{\hbar}[\hat H,\hat\rho]
-\frac{i}{2\hbar}\alpha[\{\hat q,\hat p\},\hat\rho]
-\frac{i\eta}{2\hbar}[\hat q,\{\hat p,\hat\rho\}] +\\[6pt]
&-\frac{D_{pp}}{\hbar^2}[\hat q,[\hat q,\hat\rho]]-\frac{D_{qq}}{\hbar^2}[\hat p,[\hat p,\hat\rho]]+\frac{2D_{qp}}{\hbar^2}[\hat q,[\hat p,\hat\rho]]
\, .
\end{aligned}
\end{equation}

Here, $\beta$, $\alpha$, $\eta$ and $D_{qp}$ are free parameters, while $D_{pp}$ and $D_{qq}$ are not independent but are fixed to be
\begin{equation}
    D_{pp}= m\eta/\beta \,,\quad D_{qq} = \biggl(\dfrac{\hbar^2\eta \beta}{16m}+\dfrac{\beta D_{qp}^2}{\eta m}\biggr)\,.
\end{equation}
 In order to get a translation-covariant (TC) dynamics, apart from possible symmetry breaking terms in $H$, we must impose $\alpha = 0$ to cancel the extra Hamiltonian term $[\{\hat q,\hat p\},\hat\rho]$ in~\cref{eq: trasl_cov_me}. We now move to the phase-space Wigner representation of the dynamics using the mapping rules summarized in Appendix \ref{app: conversions}. In this way we obtain the following Fokker-Planck equation:

\begin{equation}\label{eq: TC_FP}
    \begin{aligned}
        &\frac{\partial W_{\hat{\rho}}}{\partial t} = \bigl\{W_{\hat{H}},W_{\hat{\rho}}\bigr\}_{*} + \eta\frac{\partial}{\partial p}(pW_{\hat{\rho}}) + D_{pp}\frac{\partial^2}{\partial p^2}W_{\hat{\rho}}+\\
        &+D_{qq}\frac{\partial^2}{\partial q^2}W_{\hat{\rho}}+2D_{qp}\frac{\partial^2}{\partial p\partial q}W_{\hat{\rho}}\,.
    \end{aligned}
\end{equation}

Note that, in the Gaussian regime we are considering, the Moyal brackets $\{W_{\hat{A}},W_{\hat{B}}\}_{*}$, that originate from the unitary part of the evolution \cite{baker1958formulation, zachos2005quantum}, reduce to the classical Poisson brackets. 

In particular, consider the case of a free particle, for which the stationary solution is unique and is the expected Gibbs state $\hat{\rho}_{th}=\exp(-\beta\hat p^2/2m)/\calZ$ (with $\calZ$ the partition function) with this choice of parameters. It is particularly easy to see that the steady-state cannot satisfy DB for this equation provided $D_{qp}\neq 0$, despite it being the canonical thermal state. To this end, it is sufficient  to consider the diffusion matrix
 \begin{equation}
B=\begin{pmatrix}
  2D_{qq} & 2D_{qp}\\ 
  2D_{qp} &  2D_{pp}
\end{pmatrix}\,,
 \end{equation}
which violates the second condition in \cref{eq: cond_DB_main}, ${\bm \calE}{\bm B}{\bm \calE}^T = {\bm B}$. Indeed, as the diffusion matrix is constant, such requirement may hold only when the off--diagonal elements $D_{qp}$ are null. These elements of $B$ actually contributes to the reversible part of the current, entering the second term of the entropy flux rate:
\begin{equation}\label{eq: entropy flux rate TC}
    \Phi = -\int \,dq\,dp  \,  \biggl[\frac{\eta}{D_{pp}}pJ_p + 2D_{qp}\frac{\partial_q W_{\hat{\rho}}
    \partial_p W_{\hat{\rho}}}{W_{\hat{\rho}}}\biggr]\,,
\end{equation}
found following the procedure highlighted in \cref{app: DB theorem}. The entropy production rate instead reads:
\begin{equation}\label{eq: entropy prod rate TC}
    \Pi{=}\int \,dq\,dp \left[\sum_{k=q,p}\frac{J_k^2}{D_{kk}W_{\hat{\rho}}}\right].
\end{equation}
For the free particle, the irreversible component of the current is $J^I(q,p) = (J_q,J_p)^T$ with
\begin{equation}
\begin{aligned}
    &J_q(q,p)=-D_{qq}\partial_q W\,, \\
    &J_p(q,p)=-\eta p -D_{pp}\partial_p W\,,
\end{aligned}
\end{equation}
 
which go to zero in the steady state. Thus, both the entropy production and entropy flux rates go to zero in the steady state (note that $\partial_q W_{S} = 0$). So, although no dissipation is present, the stationary state is still driven out-of-equilibrium by another physical mechanism, related to the kinetic part of the dynamics, which cannot be captured by the entropic analysis. This is usually associated to dynamical activity in active matter systems and other nonequilibrium phenomena, and is referred to as {\it frenesy} in the literature  (cf. Ref.~\cite{maes2020frenesy, basu2015nonequilibrium} for an overview). Therefore, in order to obtain a master equation that leads the free particle to an equilibrium steady state, it is enough to set $D_{qp}=0$. In this way, one gets the so called 'minimally invasive' extension of the original Caldeira-Leggett master equation \cite{breuer2002theory}, which was introduced as a CPTP version of the latter. 

However, the fact that Eq. (\ref{eq: TC_FP}) with $D_{qp}=0$ leads to correct thermalization is just a peculiarity of the free particle. As we show now, as soon as one adds a potential, a constant entropy production arises in the steady state. Indeed, for this model the steady state is $W_S(q,p)= \frac{1}{2\pi \sqrt{\det V_S}} \exp\left( -\frac{1}{2} (q,p) V_S^{-1} (q,p)^T \right)$, $V_S$ being the covariance matrix. In the case of a harmonic oscillator, the entries of $V_S$ are

\begin{equation}
\small
    \begin{aligned}
        \sigma^2_{q,S} &= <q^2>_{S}= \dfrac{\eta D_{qq}}{ \omega^2} + \dfrac{D_{qq}}{\eta} + \dfrac{D_{pp}}{\eta(m\omega)^2},\,
        \\
        \sigma_{pq,S} &= \frac{1}{2}<\{q,p\}>_{S} = -mD_{qq},\\
        \sigma^2_{p,S} &= <p^2>_{S} = \dfrac{(m\omega)^2D_{qq}}{\eta} + \dfrac{D_{pp}}{\eta}.
    \end{aligned}
\end{equation}
The entropy production remains functionally the same as in \cref{eq: entropy prod rate TC}, but plugging this new steady state in that expression one readily finds: 
\begin{equation}
\label{Piinfty}
    \Pi_S
    =\frac{D_{qq}}{D_{pp}}\eta m^{2}\omega^{2}\,.
\end{equation}
 The physical mechanism that drives the system out-of-equilibrium in this case is a dissipation that arises due to the diffusion term in the $q$ direction that is not balanced by any corresponding friction, thus not satisfying a fluctuation--dissipation relation. The reason why this does not affect the free particle is that in that case the Wigner function associated to the steady state  has no position dependence.

As a final remark, we note that the term that causes the diffusion in the position variable is exactly the term introduced to make the Caldeira-Leggett master equation CPTP. If one were to remove this term, giving up on the CPTP requirement, one would recover the standard Klein-Kramers equation describing the phase-space dynamics of a classical Brownian particle \cite{kramers1940brownian}, which is thus known to have the desired thermodynamic properties \cite{tome2010entropy}. 

\section{Breaking the TC assumption}
\label{sec:brokenTC}
In fact, it is indeed possible to get a CPTP model that satisfies DB if one is allowed to include in the CPTP Caldeira-Legget model an Hamiltonian term breaking the TC symmetry, arising from the interaction with the environment \cite{chen2025thermodynamically, nicacio2024complete}. This amounts to setting $D_{qp}=0$ and $\alpha \neq 0$ in \cref{eq: trasl_cov_me}. The new term leads to a squeezing of the state during the dynamics, which in general will lead to a steady--state with non--zero correlation. However, we will show that the squeezing must be fine-tuned to compensate the position diffusion term to achieve equilibrium. This makes the presence of the squeezing term in the Hamiltonian quite {\it ad-hoc}. Considering again the harmonic oscillator, one has the following drift and diffusion matrix:
\begin{equation}\label{eq:A B matrices}
    A=\begin{pmatrix}
  \alpha & \frac{1}{m}\\ 
  -m\omega^2 & -(\alpha +\eta)
\end{pmatrix}\,,\,\,\,\,\,\,\\\,B=\begin{pmatrix}
  2D_{qq} & 0\\ 
  0 &  2D_{pp}
\end{pmatrix}\,.
\end{equation}
Using these in \cref{eq: cond_DB_main}, we find that the steady-state Wigner function needs to be of the following form to satisfy the DB property:
\begin{equation}\label{eq: alfa dist}
    W_S(q,p)=\calN \exp\biggl( \frac{\alpha}{2D_{qq}} q^2 -\frac{\eta+\alpha}{2D_{pp}}p^2\biggr)\,,
\end{equation}
from which immediately follows $\alpha<0$ and $|\alpha|<\eta$ as conditions for $W_S$ to be normalizable. Indeed, substituting \cref{eq: alfa dist} in the stationary Fokker-Planck equation one finds that the $\alpha$ that solve the stationary problem is
\begin{equation}
    \alpha_{DB} = -\eta\frac{D_{qq}m^2\omega^2}{D_{pp}+D_{qq}m^2\omega^2}\,,
\end{equation}
So, in summary given this dynamics, with $\alpha=\alpha_{DB}$, the steady--state satisfies DB and has the following inverse covariance matrix 
\begin{equation}
V_{S}=\begin{pmatrix}
   -\frac{D_{qq}}{\alpha_{DB}}& 0\\ 
  0 &\frac{D_{pp}}{\eta + \alpha_{DB}}
\end{pmatrix}.
\end{equation}

We want to emphasize how  $\alpha_{DB}$ need a subtle fine tuning to guarantee DB in the steady--state. In particular, with respect to the usual fluctuation--dissipation relation $D_{pp}=m\eta/\beta$, $\alpha_{DB}$ contains a dependence on the oscillator's frequency $\omega$.

Another way to restore DB in the CPTP Caldeira-Leggett model is to include a friction term in the $q$ direction $\frac{i\xi}{2\hbar}[\hat p,\{\hat q, \hat \rho\}]$ (equivalent to a drift term $\xi \partial_q(qW(q,p))$ in phase-space) to balance the anomalous diffusion. This is a natural choice as it mimic the standard friction term in the $p$ direction, and it keeps the dynamics CPTP, as shown in Ref. \cite{nicacio2024complete}. However, it is easy to show that this case is not really different from the one we just studied. To see this, we note that this extra term can be rewritten as 
\begin{equation}
    \frac{i\xi}{2\hbar}[\hat p,\{\hat q,\hat\rho\}]
=
\frac{i\xi}{2\hbar}[\{\hat q,\hat p\},\hat\rho]
-
\frac{i\xi}{2\hbar}[\hat q,\{\hat p,\hat\rho\}],
\end{equation}
which means that it can be absorbed in the Hamiltonian term proportional to $\alpha$ and in the standard friction term proportional to $\eta$ (respectively the second and the third terms in the RHS of \eqref{eq: trasl_cov_me}). Also here for DB to be satisfied the parameter $\xi$ needs to be fine tuned to the system and bath parameters.

\section{Conclusions} 
\label{sec:concl}
We have shown that, in general, a Markovian and Gaussian Master equation describing QBM cannot lead to proper thermalization in the sense that the DB is always violated, if translation covariance in space of the interaction with the bath is assumed. This is a relevant conclusion as it affects many phenomenological model based on these master equations. 

This lack of thermodynamic equilibrium is due to the effects of anomalous currents in phase-space which are necessarily there when CPTP of the dynamics is imposed that drive the system out of equilibrium, either via an indefinite heating mechanism (no dissipation) or in light of a dynamical, \emph{i.e.} time-reversible, nonequilibrium contribution. We have also shown how breaking the TC symmetry may allow for master equations with an associated equilibrium state that satisfies micro--reversibility and characterized the parameters of the model in this case. However, this can be achieved only by fine tuning the parameters of the model, with respect with the Hamiltonian of the system. In this sense this added symmetry--breaking term is quite \emph{ad-hoc} and should not be regarded as a physically natural solution of the problem.

Future extensions of this work could investigate how these conclusions are affected by models that depart from the Markovian and Gaussian assumptions adopted here. Exploring non-Gaussian dynamics might be useful in the derivation of a QBM theory that is consistent from both a thermodynamic and quantum mechanical point of view. However, in order to study such models, one would need to resort to the fully quantum generalization of DB~\cite{carmichael1976detailed}, and  generalize the construction of the entropy production rate to states with negative Wigner functions as well. This could be done, for example, using the Wehrl entropy \cite{zicari2023role, artini2025non}, which is always well defined, thus fully leaving the classical framework (see Appendix \ref{app: wherl and Vn}). The assessment of non-Markovian effects is a second important line of investigation. A potential approach would be to consider non-Markovian master equations derived from the microscopic dynamics without approximation, such as the one introduced by  Hu, Paz and Zhang~\cite{hu1992quantum}, whose implications in the context addressed in this work remain fully unexplored.

\textit{Acknowledgements -}
We acknowledge support from
the European Union’s Horizon Europe EIC-Pathfinder
project QuCoM (101046973), the Royal Society Wolfson Fellowship (RSWF/R3/183013), the UK EPSRC (EP/X021505/1), the Department for the Economy of Northern Ireland under the US-Ireland R\&D Partnership Programme, the ``Italian National Quantum Science and Technology Institute (NQSTI)" (PE0000023) - SPOKE 2 through project ASpEQCt,  the Italian Ministry of University and Research under PNRR - M4C2-I1.3 Project PE-00000019 "HEAL ITALIA" (CUP B73C22001250006) and the PNRR Project QUANTIP – Partenariato Esteso NQSTI – PE00000023 – Spoke
9 (CUP E63C22002180006). SD acknowledges support from Istituto Nazionale di Fisica Nucleare (INFN).

\appendix

\section{Detailed Balance and stochastic thermodynamics}\label{app: DB theorem}
We construct the entropy production rate starting from a general Fokker-Planck equation and leveraging arguments based on time-reversal symmetries. We also show the connection between this quantity and micro-reversibility and identify the cases when it ceases to be the sole quantity that characterizes deviations from equilibrium.

Consider the general time-homogeneous Fokker-Planck equation for the evolution of the probability density function of a $N-$dimensional Markovian stochastic process without jumps
\begin{equation} \label{eq: general_FP}
 \partial_t P(\bm x,t)= \left(-\nabla\!\cdot\! {{\boldsymbol{A}}}(\boldsymbol{x})
    +\frac{1}{2} \nabla\!\cdot\!{\bm B}(\bm x)\nabla\right)P(\bm x,t) \,,
\end{equation}
where ${\bm A}(\bm x)$ is the drift vector and ${\bm B}(\bm x)$ is the diffusion matrix and assume there exists a  stationary solution $P_S(\bm x)$ for such a process. Assume that these variables transform under the action of the time-reversal operator ${\bm \calE}= \textit{diag}(\epsilon_1,...,\epsilon_N)$ as $x_i \rightarrow \epsilon_i x_i$, $\epsilon_i = \pm 1$, $i=1,..,N$. Typically, one deals with position and momentum variables for which $\epsilon_i = 1$ and $\epsilon_j = -1$, respectively. Then, the stationary state is said to satisfy the DB \textit{condition} if the micro-reversibility relation holds
\begin{equation}\label{eq: def_DB}
    P(\bm x,\tau|\bm x',0)P_S(\bm x') = P(\bm \calE\bm x',\tau|\bm \calE \bm x,0)P_S(\bm x) \,,
\end{equation}
$P(\bm x,\tau|\bm x',0)$ being the conditional probability of finding the system in the state $x$ at time $t$, given it was observed in position $x'$ at the initial time. It can be shown (cf. Ref.~\cite{gardiner1985handbook}, for instance), that the necessary and sufficient conditions for DB to hold are
\begin{equation} \label{eq: cond_DB}
    \begin{cases}
        & \!\!\!\!\!\!\!\!{\bm \calE}{\bm A}(\bm \calE\bm x)P_S(\bm x)=-{\bm A}(\bm x)P_S(\bm x)+\nabla_{\bm x}({\bm B}(\bm x)P_S(\bm x)),\\
        & \!\!\!\!\!\!\!\! {\bm \calE} {\bm B}(\bm \calE \bm x){\bm \calE}^T={\bm B}(\bm x).
    \end{cases}
\end{equation}
In the following, we restrict to the case of  ({\it i}) homogeneous diffusion matrix, {\it i.e.} ${\bm B}(\bm x)= \bm B$, with $\bm B$ a constant symmetric matrix and ({\it ii}) drift vector linear in the coordinates, ${\bm A}(\bm x) = \bm A\bm x$, the drift matrix $\bm A$ being constant.
To get a thermodynamic description of the system, we start by writing \cref{eq: general_FP} as a continuity equation. We have
\begin{equation} \label{eq: continuity}
\begin{aligned}
    & \partial_t P(\bm x,t) = -\nabla_{\bm{x}}\cdot J(\bm{x},t),\\
    & J(\bm x,t) = {\bm A}\bm x P(\bm x,t)-\frac{1}{2}{\bm B}\nabla_{\bm x}P(\bm x,t).
\end{aligned}
\end{equation}
We can split the current into its time-symmetric component ${\bm \calE} J^{I}(\calE \bm x) = J^{I}( \bm x)$ and  anti-symmetric one ${\bm \calE} J^{R}(\calE \bm x) = -J^{R}(\bm x)$ with $J( \bm x)=J^{R}( \bm x)+J^{I}( \bm x)$. We also define the reversible and irreversible components of a matrix $\bm M$ as ${\bm M}^{R} = \frac{1}{2}({\bm M}-{\bm \calE}{\bm M}{\bm \calE}^T)$ and ${\bm M}^{I} = \frac{1}{2}({\bm M}+{\bm \calE}{\bm M}{\bm \calE}^T)$. Note that by construction $\bm M^I$ is always block diagonal and $\bm M^R$ is always block anti-diagonal. 

Using this notation, the two components of the current are
\begin{equation} \label{eq: currents}
    \left\{\begin{matrix}
        J^{R}(\bm x,t) = {\bm A}^{R}\bm x P(\bm x,t)-\frac{1}{2}{\bm B}^{R}\nabla_{\bm x}P(\bm x,t) \,,\\
        J^{I}(\bm x,t) = {\bm A}^{I}\bm x P(\bm x,t)-\frac{1}{2}{\bm B}^{I}\nabla_{\bm x}P(\bm x,t) \,.
    \end{matrix}\right.
\end{equation}
\newline
We are interested in the rate $\dot S(t)$ of the Shannon entropy of the system, where $S(t) = -\int \,d\bm x P(\bm x,t)\ln(P(\bm x,t))$. Using \cref{eq: continuity} and integrating by parts (assuming vanishing probability distributions and currents at spatial infinity) we have
\begin{align}
     \frac{dS}{dt} &= -\int d\bm x J(\bm x)\cdot \nabla_{\bm x}\ln(P(\bm x)) \\
    &\! =  -\!\int d\bm x J^{R}(\bm x)\cdot\frac{\nabla_{\bm x}P(\bm x)}{P(\bm x)}-\int d\bm x  J^{I}(\bm x)\cdot\frac{\nabla_{\bm x}P(\bm x)}{P(\bm x)},\nonumber
\end{align}
where, on the right-hand side, the dependence on $t$ is left implicit. The reversible part of the drift matrix does not actually contribute to determining the Shannon entropy rate. In fact, inserting the definition of $J^R(\boldsymbol{x})$ in the first integral, we can observe that:
\begin{equation}
    \int d \boldsymbol{x}\nabla_{\boldsymbol{x}}P(\boldsymbol{x})\boldsymbol{A}^R\boldsymbol{x}=-\Tr \boldsymbol{A}^R=0\,.
\end{equation}
Here, after integrating by part, we leveraged that $\partial_ix_j=\delta_{ij}$ and that the reversible drift matrix is traceless by construction.
In the second integral, we perform the substitution $\nabla_{\bm x}P = 2P({\bm B}^{I})^{-1}{\bm A}^{I}\bm x - 2({\bm B}^{I})^{-1}J^{I}$ (which follows from the second equation in \eqref{eq: currents}), finally getting the following expression of the rate of the Shannon entropy
\begin{widetext}
    \begin{equation}\label{eq: general entropy}
\begin{aligned}
     \frac{dS}{dt} = &-2\int \,d\bm x \, ({\bm B}^{I})^{-1}{\bm A}^{I}\bm x\cdot J^{I}(\bm x) +\frac{1}{2}\int \,d\bm x \,\, \frac{\nabla_{\bm x}P(\bm x){\bm B}^{R}\nabla_{\bm x}P(\bm x)}{P(\bm x)}+ 2\int \,d\bm x \,\, \frac{J^{I}(\bm x)\,({\bm B}^{I})^{-1}J^{I}(\bm x)}{P(\bm x)}\,.
\end{aligned}
\end{equation}
\end{widetext}

We now claim and prove the following theorem concerning the relation between micro-reversibility and the irreversible component of the current.

\begin{app}
For a homogeneous Fokker-Planck equation such that %whose diffusion matrix satisfies the condition 
    ${\mathcal{E}}{\bm B}{\mathcal{E}}^T={\bm B}$, if a stationary solution $P_S(\bm x)$ exists and is such that the irreversible component of the current is zero, then the DB conditions hold. 
%    Consider a homogeneous Fokker-Planck equation written as in \cref{eq: continuity}. If a stationary solution $P_S(\bm x)$ exists and is such that the irreversible component of the current, as defined in \cref{eq: currents}, is zero and assuming ${\bm B}={\bm B}^{I}$, then the Detailed Balance conditions (\ref{eq: cond_DB}) hold. 
\end{app}

\begin{proof}
    Using the definition of the irreversible current \begin{align}
        J_S^{I} = 0 
        \iff 
        {\bm A}^{I}\bm x P_S(\bm x)-\frac{1}{2}&{\bm B}^{I}\nabla_{\bm x}P_S(\bm x) = 0 \\
       &\iff \nonumber \\
        ({\bm A} + {\bm \calE}{\bm A}{\bm \calE}^T)\bm x P_S(\bm x)-\frac{1}{2}({\bm B} &+ {\bm \calE}{\bm B}{\bm \calE}^T)\nabla_{\bm x}P_S(\bm x) = 0 \nonumber\\
   &\iff \nonumber\\
         {\bm \calE}{\bm A}{\bm \calE}^T\bm x P_S(\bm x)=-{\bm A}\bm x P_S(\bm x)&+\frac{1}{2}({\bm B} + {\bm \calE}{\bm B}{\bm \calE}^T)\nabla_{\bm x}P_S(\bm x) \nonumber
    \end{align}
    Using the assumption ${\bm B} = {\bm \calE}{\bm B}{\bm \calE}^T$, which is the second line of the DB conditions written in the matrix form as in \cref{eq: cond_DB}, we recover exactly the first line of the same set of conditions, proving the statement.
\end{proof}

Assuming that the diffusion matrix has no reversible component, the second term in the RHS side of \cref{eq: general entropy} disappears, while the other two terms are both zero if and only if the micro-reversibility conditions \cref{eq: cond_DB} holds true, as stated by the Theorem. Since the last one is always non-negative \footnote{A sufficient condition is that ${\bm B}^{I}$ is positive semidefinite. This is the case for physical systems since for $N$ particles $B^I$ is block diagonal with positive semidefinie blocks to ensure non-negative diffusion coefficients.}, we can thus consider it as a measure of the distance from equilibrium and we call it entropy production rate, here denoted as $\Pi$, in analogy with the usual Second law of Thermodynamics. This identification is also grounded in Fluctuation Theorems, which link this quantity to the average of the entropy production rate of stochastic trajectories \cite{seifert2012stochastic, santos2017wigner}. The remainder of the right-hand side of \cref{eq: general entropy} is thus identified with the opposite of the entropy flux rate between the system and the environment, which can be positive or negative, called $\Phi$ in our notation. In summary, we are able to write the rate of entropy of the system in terms of two contributions:
\begin{equation}
    \frac{dS}{dt}(t) = \Pi (t) - \Phi (t) \,,
\end{equation}
one of which, $\Pi$, is identified with the global entropy production rate. Note that at the stationary state $\Phi = \Pi$, but each contribution is equal to zero if and only if the DB condition holds \footnote{DB implies $\Phi = 0$, but the converse is not true.}, meaning that the system is in equilibrium. When this condition is not verified, the system is in a non-equilibrium steady state (NESS), characterized by a constant entropy production (and flux) rate $\Pi_{\infty}>0$.

On the other hand, if there is a reversible component in the diffusion term, a zero irreversible current in the steady state does not imply the detailed balance conditions, meaning that in this scenario the dissipation quantified by the entropy production rate is not the only physical mechanism keeping the system out of equilibrium. Indeed, in this scenario an out-of-equilibrium steady state with $\Pi = 0$ can be found, as shown for the Translation-covariant master equation of a free particle in the main text (\cref{eq: trasl_cov_me}). In the context of classical statistical mechanics, these reversible, cross-diffusion terms can arise due to an anisotropy in the Brownian particle, e.g. an ellipsoidal particle \cite{han2006brownian}, or due to an anisotropic environment, and they are also found in plasma physics \cite{dougherty1964model}.

\section{Going from the master equation to the Fokker-Planck equation} \label{app: conversions}

We report in the table the conversion rules for the most common quadratic terms when passing from the master-equation picture to the phase-space representation.
\begin{table}[h]\label{tab: transl}
\centering
\renewcommand{\arraystretch}{2}
\begin{tabular}{|c|c|}
\hline
\textbf{Master Equation}               & \textbf{Fokker-Planck Equation}               \\ \hline
$-\frac{i}{\hbar}[\hat{H},\hat{\rho}]$ & $\bigl\{W_{\hat{H}},W_{\hat{\rho}}\bigr\}_*$ \\ \hline 
$[\hat{q},\{\hat{p},\hat{\rho}\}]$     & $2i\hbar \partial_p (pW_{\hat{\rho}})$        \\ \hline
$[\hat{p},\{\hat{q},\hat{\rho}\}]$     & $-2i\hbar \partial_q (qW_{\hat{\rho}})$       \\ \hline
$[\hat{p},[\hat{p},\hat{\rho}]]$       & $-\hbar^2\partial_q^2W_{\hat{\rho}}$           \\ \hline
$[\hat{q},[\hat{q},\hat{\rho}]]$       & $-\hbar^2\partial_p^2W_{\hat{\rho}}$           \\ \hline
$[\hat{q},[\hat{p},\hat{\rho}]]$ or $[\hat{p},[\hat{q},\hat{\rho}]]$    & $\hbar^2\partial_q \partial_p$                                         \\ \hline 
\end{tabular}
\end{table}

with the Moyal parenthesis $\{\cdot ,\cdot\}_{*}$ in the first line defined as:
\begin{equation}
\{A,B\}_{*}
=
\frac{2}{\hbar}\,
A(q,p)\,
\sin\!\left[
\frac{\hbar}{2}
\left(
\overleftarrow{\partial}_{q}\overrightarrow{\partial}_{p}
-
\overleftarrow{\partial}_{p}\overrightarrow{\partial}_{q}
\right)
\right]
B(q,p).    
\end{equation}

\section{On different choices of entropy production}
\label{app: wherl and Vn}

In the following, we address two other definitions of entropy production. The first one, the one that is derived from the Wehrl entropy, is shown to be equivalent to the Wigner entropy in this context. The second one, Spohn's entropy production, is argued to be inadequate at addressing the entropy of quantum brownian particles. This analysis solidifies our adoption of the Wigner entropy in the main text.

We first focus on the $Q-$representation and the Wehrl entropy defined as: $S_{Q} = -\iint dq\,dp\, Q(q,p)\ln(Q(q,p))$, where $Q(\alpha) = \frac{1}{\pi}\bra{\alpha}\hat \rho\ket{\alpha}$. The Wigner function and the $Q-$representation are connected via a convolution with the Gaussian with minimal variance allowed by the Uncertainty principle \cite{walls2008quantum}:
\begin{equation}
    \begin{aligned}
        & Q(q,p) = \iint dq'\,dp'\, W(q',p') G(q-q',p-p') \,,\\
        & G(x,y) = \frac{1}{\hbar \pi}\exp(-\frac{1}{\hbar}(m\omega x^2+p^2/m\omega))\,.
    \end{aligned}
\end{equation}
Thus, given a QFP equation in terms of the Wigner function of the form $\partial_t W = \calL^W [W]$, we get the corresponding equation of motion in the $Q-$representation just by applying the convolution on both sides:
\begin{equation*}
        \partial_t Q = \calL^Q [Q] = \iint dq'\,dp'\, \calL^W [W(q',p')] G(q-q',p-p')\,.
\end{equation*}
It is easy to prove that for all terms in the Liouvillian $\calL^W$ that are derivatives of $W$ alone \footnote{Considering an harmonic oscillator this holds for the unitary part in particular.}, it is enough to substitute $W\rightarrow Q$ in the expression, whereas the friction term yields: 
\begin{equation}
        \frac{\partial (p W)}{\partial p}\rightarrow \frac{\partial}{\partial p} \biggl(p + \frac{\hbar m\omega}{2}\frac{\partial}{\partial p}\biggr)Q(q,p) \,.
\end{equation}
Therefore, moving to the $Q-$representation of the QFP equations considered in this Letter amounts to shifting the momentum diffusion constant by $ \dfrac{\hbar m\omega\eta}{2}$. So, for example, the QFP equation with translational-covariant dissipator (Eq. \eqref{eq: trasl_cov_me} and Eq. \eqref{eq: TC_FP} in the main  text) for $Q$ considering at most quadratic Hamiltonians:
\begin{equation*}
\begin{aligned}
    &  \frac{\partial Q_{\hat{\rho}}}{\partial t} = \bigl\{Q_{\hat{H}},Q_{\hat{\rho}}\bigr\}_{PB} + \eta\frac{\partial(p\,Q_{\hat{\rho}})}{\partial p} +\\
    & + \biggl(D_{pp}+\frac{\hbar m\omega\eta}{2}\biggr)\frac{\partial^2Q_{\hat{\rho}}}{\partial p^2}+D_{qq}\frac{\partial^2 Q_{\hat{\rho}}}{\partial q^2} + 2D_{pq}\frac{\partial^2 Q_{\hat{\rho}}}{\partial p \partial q}\,,
\end{aligned}
\end{equation*}
where $Q_{\hat{H}}$ is the Hamiltonian in the $Q-$representation and $\{\cdot,\cdot\}_{PB}$ are the classical Poisson brackets. Considering the CL model (putting $D_{qq}=D_{qp}=0$), allow for a simple discussion about the DB in this picture. At first glance, it would seem that with the new diffusion constant $\tilde D_{pp}=D_{pp}+\frac{\hbar m\omega\eta}{2}$ together with Einstein relation $D_{pp} = \frac{m\eta}{\beta}$, that usually ensures DB, the conditions \cref{eq: cond_DB_main} are violated if one naively substitutes $W\rightarrow Q$. This apparent inconsistency arises because one needs to express those conditions in the $Q$-representation as well. By doing so, one finds that the second line in \cref{eq: cond_DB_main} remains unchanged, while the first one becomes (for this model $A$ and $B$ are given in \cref{eq:A B matrices} with $\alpha=D_{qq}=0$):
\begin{equation}\label{eq: QDB}
    -\eta \biggl(p+\frac{\hbar m\omega}{2}\frac{\partial}{\partial_p}\biggr)Q_S = D_{pp}\frac{\partial}{\partial_p}Q_S \,.
\end{equation}
The distribution satisfying \cref{eq: QDB} and solving the stationary problem of the FP equation is (considering again the harmonic oscillator):
\begin{equation}
    Q_{S}(q,p) = \calN \exp(-\biggl(\frac{p^2}{\frac{2m}{\beta} + \hbar m\omega} + \frac{m^2\omega^2 q^2}{\frac{2m}{\beta} + \hbar m\omega} \biggl))\,,
\end{equation}
which is exactly the $Q-$function of the desired equilibrium Boltzmann distribution. Furthermore, having the two QFP equations the exact same functional structure, the framework used to define the entropy production and flux rates can be carried out in the exact same way, leading to the same conclusions regarding the thermodynamic behavior of the system.

This discussion seems to suggest that the classical DB conditions need to be adapted depending on which mathematical description of the quantum dynamics is chosen, with the exception of the Wigner picture which retains the same structure of the classical case, as long as Gaussian and Markovian dynamics are considered.

Finally, we discuss the choice of Von Neumann entropy as the main figure of merit and the Spohn's definition of entropy production \cite{spohn1978entropy} that usually follows from this choice. The argument against the adoption of this definition is the so called 'ultracold catastrophe', which is a divergence of the entropy production in the limit of a zero temperature heat bath \cite{uzdin2021passivity}. This non-physical feature suggests that a different definition of entropy is necessary to tackle brownian particles in the quantum regime.

\section{Relation with the KMS Detailed Balance}\label{app: KMS}

There is no universally accepted quantum generalization of Detailed Balance and thus different definitions maybe found in the literature \cite{fagnola2009two}. One of the most widely used is the Kubo-Martini-Schwinger (KMS) Detailed Balance, also thanks to its physical connection to the Petz recovery map \cite{scandi2025thermalization, fagnola2010generators}. This definition is based on the properties of the dual of the generator of the dynamics with respect to the inner product on the algebra of observables:
\begin{equation}
    \langle A,B \rangle_{\Sigma} = \Tr[\sqrt{\Sigma} A^\dagger \sqrt{\Sigma} B]\,,
\end{equation}
where $\Sigma$ is a suited reference state. This is the definition of DB used in Refs. \cite{chen2025thermodynamically, nicacio2024complete, stockburger2017thermodynamic} for deriving thermodinamically consistent Gaussian GKSL equations. 

One can prove that the KMS DB condition can be written in terms of the the drift matrix $A$, the diffusion matrix $B$ and the steady state's covariance matrix $V_{\Sigma}$ as:
\begin{equation}
    B = -2 A_{irr}V_{\Sigma}.
\end{equation}

It is possible to show that this condition is equivalent to our first detailed-balance condition in \cref{eq: cond_DB_main} when one chooses \(\Sigma = V_S\), where \(V_S\) is the covariance matrix of the steady state. However, \cref{eq: cond_DB_main} also requires
${\bm \calE}{\bm B}{\bm \calE}^T = {\bm B}$,
which implies that micro-reversibility is automatically broken whenever \(D_{qp}\neq 0\). Aside from this point, for all the other models considered---where \(D_{qp}\) was already set to zero---this second condition plays no relevant role. Therefore, in those cases, our notion of DB and the KMS one are in fact fully equivalent for markovian, gaussian dynamics. This shows that the conclusions we reached remain robust under a change in the definition of the DB condition.


\begin{thebibliography}{99}

\bibitem{gorini1976completely}
V.~Gorini, A.~Kossakowski, and E.~C.~G. Sudarshan.
\newblock Completely positive dynamical semigroups of {N}-level systems.
\newblock {\em Journal of Mathematical Physics}, 17(5):821--825, 1976.

\bibitem{lindblad1976generators}
G.~Lindblad.
\newblock On the generators of quantum dynamical semigroups.
\newblock {\em Communications in Mathematical Physics}, 48(2):119--130, 1976.

\bibitem{homa2019positivity}
G.~Homa, J.~Z. Bern{\'a}d, and L.~Lisztes.
\newblock Positivity violations of the density operator in the {Caldeira-Leggett} master equation.
\newblock {\em The European Physical Journal D}, 73(3):53, 2019.

\bibitem{breuer2002theory}
H.-P. Breuer and F.~Petruccione.
\newblock {\em The theory of open quantum systems}.
\newblock OUP Oxford, 2002.

\bibitem{LandiPaternostro}
G.~T. Landi and M.~Paternostro.
\newblock Irreversible entropy production: From classical to quantum.
\newblock {\em Rev. Mod. Phys.}, 93:035008, 2021.

\bibitem{einstein1905movement}
A.~Einstein.
\newblock On the movement of small particles suspended in stationary liquids required by the molecular-kinetic theory of heat.
\newblock {\em Ann. d. Phys}, 17:549, 1905.

\bibitem{kubo1966fluctuation}
R.~Kubo.
\newblock The fluctuation-dissipation theorem.
\newblock {\em Reports on Progress in Physics}, 29(1):255, 1966.

\bibitem{sato1984operator}
K.-i. Sato and M.~Yamazato.
\newblock Operator-selfdecomposable distributions as limit distributions of processes of {O}rnstein-{U}hlenbeck type.
\newblock {\em Stochastic Processes and their Applications}, 17(1):73--100, 1984.

\bibitem{barrera2021cutoff}
G.~Barrera, M.~A. H{\"o}gele, and J.~C. Pardo.
\newblock Cutoff thermalization for {Ornstein--Uhlenbeck} systems with small {{L}{\'e}vy} noise in the {{W}asserstein} distance.
\newblock {\em Journal of Statistical Physics}, 184(3):27, 2021.

\bibitem{gardiner1985handbook}
C.~W. Gardiner.
\newblock Handbook of stochastic methods for physics, chemistry and the natural sciences.
\newblock {\em Springer Series in Synergetics}, 1985.

\bibitem{vacchini2002quantum}
B.~Vacchini.
\newblock Quantum optical versus quantum {Brownian} motion master equation in terms of covariance and equilibrium properties.
\newblock {\em Journal of Mathematical Physics}, 43(11):5446--5458, 2002.

\bibitem{ramezani2018quantum}
M.~Ramezani, F.~Benatti, R.~Floreanini, S.~Marcantoni, M.~Golshani, and A.~Rezakhani.
\newblock Quantum detailed balance conditions and fluctuation relations for thermalizing quantum dynamics.
\newblock {\em Physical Review E}, 98(5):052104, 2018.

\bibitem{tome2010entropy}
T.~Tom{\'e} and M.~J. de Oliveira.
\newblock Entropy production in irreversible systems described by a {Fokker-Planck} equation.
\newblock {\em Physical Review E—Statistical, Nonlinear, and Soft Matter Physics}, 82(2):021120, 2010.

\bibitem{kramers1940brownian}
H.~A. Kramers.
\newblock Brownian motion in a field of force and the diffusion model of chemical reactions.
\newblock {\em Physica}, 7(4):284--304, 1940.

\bibitem{seifert2008stochastic}
U.~Seifert.
\newblock Stochastic thermodynamics: principles and perspectives.
\newblock {\em The European Physical Journal B}, 64:423--431, 2008.

\bibitem{onsager1931reciprocal}
L.~Onsager.
\newblock Reciprocal relations in irreversible processes. {I}.
\newblock {\em Physical Review}, 37(4):405, 1931.

\bibitem{caldeira1981influence}
A.~O. Caldeira and A.~J. Leggett.
\newblock Influence of dissipation on quantum tunneling in macroscopic systems.
\newblock {\em Physical Review Letters}, 46(4):211, 1981.

\bibitem{hu1992quantum}
B.~L. Hu, J.~P. Paz, and Y.~Zhang.
\newblock Quantum {Brownian} motion in a general environment: Exact master equation with nonlocal dissipation and colored noise.
\newblock {\em Physical Review D}, 45(8):2843, 1992.

\bibitem{hu1993quantum}
B.~Hu, J.~P. Paz, and Y.~Zhang.
\newblock Quantum {Brownian} motion in a general environment. {II}. {Nonlinear} coupling and perturbative approach.
\newblock {\em Physical Review D}, 47(4):1576, 1993.

\bibitem{seifert2012stochastic}
U.~Seifert.
\newblock Stochastic thermodynamics, fluctuation theorems and molecular machines.
\newblock {\em Reports on Progress in Physics}, 75(12):126001, 2012.

\bibitem{van2013stochastic}
C.~Van~den Broeck.
\newblock Stochastic thermodynamics: A brief introduction.
\newblock In {\em Physics of Complex Colloids}, pages 155--193. IOS Press, 2013.

\bibitem{vacchini2007relaxation}
B.~Vacchini and K.~Hornberger.
\newblock Relaxation dynamics of a quantum {Brownian} particle in an ideal gas.
\newblock {\em The European Physical Journal Special Topics}, 151(1):59--72, 2007.

\bibitem{giovannetti2001phase}
V.~Giovannetti and D.~Vitali.
\newblock Phase-noise measurement in a cavity with a movable mirror undergoing quantum {Brownian} motion.
\newblock {\em Physical Review A}, 63(2):023812, 2001.

\bibitem{brunelli2018experimental}
M.~Brunelli, L.~Fusco, R.~Landig, W.~Wieczorek, J.~Hoelscher-Obermaier, G.~Landi, F.~Semi{\~a}o, A.~Ferraro, N.~Kiesel, T.~Donner, G.~De~Chiara, and M.~Paternostro.
\newblock Experimental determination of irreversible entropy production in out-of-equilibrium mesoscopic quantum systems.
\newblock {\em Physical Review Letters}, 121(16):160604, 2018.

\bibitem{brunelli2016irreversibility}
M.~Brunelli and M.~Paternostro.
\newblock Irreversibility and correlations in coupled quantum oscillators.
\newblock {\em arXiv preprint arXiv:1610.01172}, 2016.

\bibitem{santos2017wigner}
J.~P. Santos, G.~T. Landi, and M.~Paternostro.
\newblock Wigner entropy production rate.
\newblock {\em Physical Review Letters}, 118(22):220601, 2017.

\bibitem{santos2018irreversibility}
J.~P. Santos, A.~L. de~Paula~Jr, R.~Drumond, G.~T. Landi, and M.~Paternostro.
\newblock Irreversibility at zero temperature from the perspective of the environment.
\newblock {\em Physical Review A}, 97(5):050101, 2018.

\bibitem{artini2023characterizing}
S.~Artini and M.~Paternostro.
\newblock Characterizing the spontaneous collapse of a wavefunction through entropy production.
\newblock {\em New Journal of Physics}, 25(12):123047, 2023.

\bibitem{artini2025non}
S.~Artini, G.~Lo~Monaco, S.~Donadi, and M.~Paternostro.
\newblock Non-equilibrium thermodynamics of gravitational objective-collapse models.
\newblock {\em arXiv preprint arXiv:2502.03173}, 2025.

\bibitem{agarwal1973open}
G.~Agarwal.
\newblock Open quantum {Markovian} systems and the microreversibility.
\newblock {\em Zeitschrift f{\"u}r Physik A Hadrons and Nuclei}, 258(5):409--422, 1973.

\bibitem{alicki1976detailed}
R.~Alicki.
\newblock On the detailed balance condition for non-{Hamiltonian} systems.
\newblock {\em Reports on Mathematical Physics}, 10(2):249--258, 1976.

\bibitem{zicari2023role}
G.~Zicari, B.~{\c{C}}akmak, {\"O}.~E. M{\"u}stecapl{\i}o{\u{g}}lu, and M.~Paternostro.
\newblock On the role of initial coherence in the spin phase-space entropy production rate.
\newblock {\em New Journal of Physics}, 25(1):013030, 2023.

\bibitem{maes2020frenesy}
C.~Maes.
\newblock Frenesy: Time-symmetric dynamical activity in nonequilibria.
\newblock {\em Physics Reports}, 850:1--33, 2020.

\bibitem{basu2015nonequilibrium}
U.~Basu and C.~Maes.
\newblock Nonequilibrium response and frenesy.
\newblock In {\em Journal of Physics: Conference Series}, volume 638, page 012001. IOP Publishing, 2015.

\bibitem{baker1958formulation}
G.~A. Baker~Jr.
\newblock Formulation of quantum mechanics based on the quasi-probability distribution induced on phase space.
\newblock {\em Physical Review}, 109(6):2198, 1958.

\bibitem{han2006brownian}
Y.~Han, A.~M. Alsayed, M.~Nobili, J.~Zhang, T.~C. Lubensky, and A.~G. Yodh.
\newblock Brownian motion of an ellipsoid.
\newblock {\em Science}, 314(5799):626--630, 2006.

\bibitem{dougherty1964model}
J.~Dougherty.
\newblock Model {Fokker-Planck} equation for a plasma and its solution.
\newblock {\em The Physics of Fluids}, 7(11):1788--1799, 1964.

\bibitem{colla2021entropy}
A.~Colla and H.-P. Breuer.
\newblock Entropy production and the role of correlations in quantum {Brownian} motion.
\newblock {\em Physical Review A}, 104(5):052408, 2021.

\bibitem{campisi2011colloquium}
M.~Campisi, P.~H{\"a}nggi, and P.~Talkner.
\newblock Colloquium: Quantum fluctuation relations: Foundations and applications.
\newblock {\em Reviews of Modern Physics}, 83(3):771--791, 2011.

\bibitem{Imparato7a}
A.~Imparato and L.~Peliti.
\newblock The distribution function of entropy flow in stochastic systems.
\newblock {\em Journal of Statistical Mechanics: Theory and Experiment}, 2007(02):L02001, 2007.

\bibitem{zachos2005quantum}
C.~Zachos, D.~Fairlie, and T.~Curtright.
\newblock {\em Quantum mechanics in phase space: an overview with selected papers}.
\newblock World Scientific, 2005.

\bibitem{van1992stochastic}
N.~G. Van~Kampen.
\newblock {\em Stochastic processes in physics and chemistry}, volume~1.
\newblock Elsevier, 1992.

\bibitem{spohn1978entropy}
H.~Spohn.
\newblock Entropy production for quantum dynamical semigroups.
\newblock {\em Journal of Mathematical Physics}, 19(5):1227--1230, 1978.

\bibitem{uzdin2021passivity}
R.~Uzdin and S.~Rahav.
\newblock Passivity deformation approach for the thermodynamics of isolated quantum setups.
\newblock {\em PRX Quantum}, 2(1):010336, 2021.

\bibitem{scandi2025thermalization}
M.~Scandi and {\'A}.~M. Alhambra.
\newblock Thermalization in open many-body systems and {KMS} detailed balance.
\newblock {\em arXiv preprint arXiv:2505.20064}, 2025.

\bibitem{fagnola2010generators}
F.~Fagnola and V.~Umanit{\`a}.
\newblock Generators of {KMS} symmetric {Markov} semigroups on symmetry and quantum detailed balance.
\newblock {\em Communications in Mathematical Physics}, 298(2):523--547, 2010.

\bibitem{fagnola2009two}
F.~Fagnola, V.~Umanit{\`a}, et~al.
\newblock On two quantum versions of the detailed balance condition.
\newblock {\em Noncommutative Harmonic Analysis with Applications to Probability II, Banach Center Publ}, 89:105--119, 2009.

\bibitem{chen2025thermodynamically}
J.-F. Chen.
\newblock Thermodynamically {Consistent} {Lindbladians} for {Quantum} {Stochastic} {Thermodynamics}.
\newblock {\em arXiv preprint arXiv:2502.20118}, 2025.

\bibitem{nicacio2024complete}
F.~Nicacio and T.~Koide.
\newblock Complete positivity and thermal relaxation in quadratic quantum master equations.
\newblock {\em Physical Review E}, 110(5):054116, 2024.

\bibitem{stockburger2017thermodynamic}
J.~T. Stockburger and T.~Motz.
\newblock Thermodynamic deficiencies of some simple {Lindblad} operators: A diagnosis and a suggestion for a cure.
\newblock {\em Fortschritte der Physik}, 65(6-8):1600067, 2017.

\bibitem{bassi2005energy}
A.~Bassi, E.~Ippoliti, and B.~Vacchini.
\newblock On the energy increase in space-collapse models.
\newblock {\em Journal of Physics A: Mathematical and General}, 38(37):8017--8038, 2005.

\bibitem{barchielli2015quantum}
A.~Barchielli and B.~Vacchini.
\newblock Quantum {Langevin} equations for optomechanical systems.
\newblock {\em New Journal of Physics}, 17(8):083004, 2015.

\bibitem{carmichael1976detailed}
H.~Carmichael and D.~Walls.
\newblock Detailed balance in open quantum {Markoffian} systems.
\newblock {\em Zeitschrift f{\"u}r Physik B Condensed Matter}, 23(3):299--306, 1976.

\bibitem{walls2008quantum}
D.~F. Walls and G.~J. Milburn.
\newblock {\em Quantum optics}.
\newblock Springer Science \& Media, 2008.

\bibitem{roberts2021hidden}
D.~Roberts, A.~Lingenfelter, and A.~Clerk.
\newblock Hidden time-reversal symmetry, quantum detailed balance and exact solutions of driven-dissipative quantum systems.
\newblock {\em PRX Quantum}, 2(2):020336, 2021.

\bibitem{ferialdi2017dissipation}
L.~Ferialdi.
\newblock Dissipation in the {Caldeira-Leggett} model.
\newblock {\em Physical Review A}, 95(5):052109, 2017.

\bibitem{diosi1987universal}
L.~Diosi.
\newblock A universal master equation for the gravitational violation of quantum mechanics.
\newblock {\em Physics Letters A}, 120(8):377--381, 1987.

\bibitem{diosi1989models}
L.~Di{\'o}si.
\newblock Models for universal reduction of macroscopic quantum fluctuations.
\newblock {\em Phys. Rev. A}, 40(3):1165, 1989.

\bibitem{penrose1996gravity}
R.~Penrose.
\newblock On gravity's role in quantum state reduction.
\newblock {\em General Relativity and Gravitation}, 28:581--600, 1996.

\bibitem{joos1985emergence}
E.~Joos and H.~D. Zeh.
\newblock The emergence of classical properties through interaction with the environment.
\newblock {\em Zeitschrift f{\"u}r Physik B Condensed Matter}, 59(2):223--243, 1985.

\end{thebibliography}
\end{document}